\documentclass[prodmode,license]{acmsmall-ec16}
\usepackage[ruled]{algorithm2e}
\usepackage{amsmath}
\usepackage{amssymb}
\usepackage{graphicx}
\usepackage{esint}

\makeatletter

\providecommand{\tabularnewline}{\\}

\usepackage{color}
\usepackage{amsfonts}

\newcommand{\bx}{{\bf{x}}}
\newcommand{\by}{{\bf{y}}}
\newcommand{\bz}{{\bf{z}}}

\makeatother

\begin{document}

\title{Matroid Online Bipartite Matching and Vertex Cover}

\author{YAJUN WANG
\affil{Microsoft Corporation}
SAM CHIU-WAI WONG
\affil{UC Berkeley}}

\begin{abstract}
The Adwords and Online Bipartite Matching problems have enjoyed a
renewed attention over the past decade due to their connection to
Internet advertising. Our community has contributed, among other things,
new models (notably stochastic) and extensions to the classical formulations
to address the issues that arise from practical needs. In this paper,
we propose a new generalization based on matroids and show that many
of the previous results extend to this more general setting. Because
of the rich structures and expressive power of matroids, our new setting
is potentially of interest both in theory and in practice.

In the classical version of the problem, the offline side of a bipartite
graph is known initially while vertices from the online side arrive
one at a time along with their incident edges. The objective is to
maintain a decent approximate matching from which no edge can be removed.
Our generalization, called Matroid Online Bipartite Matching, additionally
requires that the set of matched offline vertices be independent in
a given matroid. In particular, the case of partition matroids corresponds
to the natural scenario where each advertiser manages multiple ads
with a fixed total budget.

Our algorithms attain the \textit{same} performance as the classical
version of the problems considered, which are often provably the best
possible. We present $1-1/e$-competitive algorithms for Matroid Online
Bipartite Matching under the small bid assumption, as well as a $1-1/e$-competitive
algorithm for Matroid Online Bipartite Matching in the random arrival
model. A key technical ingredient of our results is a carefully designed
primal-dual waterfilling procedure that accommodates for matroid constraints.
This is inspired by the extension of our recent charging scheme for
Online Bipartite Vertex Cover.

Finally, given that only few online problems were studied
in the submodular fashion, the techniques introduced in
this paper for tackling submodularity in the online setting may be
of independent interest.
\end{abstract}

\maketitle

\section{Introduction}

Recent years have seen an explosion of research on various online
matching problems thanks to Internet advertising. One prominent example
is Online Bipartite $b$-matching where advertisers wish to advertise
on a search engine platform such as Google or Bing. Each ad has a
budget which specifies the maximum number of times it should be displayed
to impressions arriving in an online fashion. Upon the arrival of
an impression, the search engine must decide on which ad to display,
if any, while respecting the budget constraint. The goal is to display
as many ads as possible to maximize revenue.

This simple setting encapsulates the essence of many online allocation
problems. Nevertheless, its lack of sophistication means that if used
directly, it could be a poor model for specific applications. Recent
research has responded by proposing variants of the problem driven
by practical needs. Examples include different stochastic models (e.g.
\cite{Feldman2009,GoelM08}) and the free disposal model \cite{feldman2009online}.
In this paper, we \textit{observe} and \textit{address} a new limitation
which may be of real world interest especially for Internet advertising.

Most of these models have the shortcoming that the budget for each
ad is specified \textit{independently}. Many advertisers, especially
those selling real products, are typically hosting multiple ads. Under
the simplistic setting above, he would be able to only budget for
each ad independently. This would not be ideal as his ads may be correlated,
meaning that his budget for one ad may depend on how much he spends
on another ad. For instance, suppose that a soft drink distributor
is advertising for both coke and sprite. He is willing to spend \$5
on coke, \$4 on sprite but only \$8 on both. Under the current framework,
this kind of preference is impossible to represent by two separate
budgets for coke and sprite.

We address this limitation by allowing an advertiser to specify the
budget $f(S)$ for \emph{each} subset $S$ of his ads, with \textit{one}
mild restriction. We require that $f$ be \textit{monotone submodular}.
The search engine would then display ads from $S$ at most a total
of $f(S)$ times. We argue that the assumption on $f$ is only mild.
Monotonicity is clearly reasonable. Submodularity also makes sense
as one should expect to observe diminishing marginal returns for the
similar ads hosted by the advertiser. Returning to our example on
coke and sprite, if \$5 is already spent on coke, our advertiser may
think that the soft drink market is more saturated than before and
hence spend less on sprite (\$3) than he otherwise would (\$4).

Readers well-versed in matroid theory would recognize that this is
equivalent to imposing a (poly)matroid constraint on the set of ads.
Somewhat surprisingly, many of the existing results in the literature
conform very well to this matroid generalization and we are able to
obtain algorithms with the same performance in this considerably more
general setting.

From a practical standpoint, we believe that this is a desirable feature
to be introduced to Internet advertising. Because of the need to budget
individually, advertisers of multiple ads may currently be under-budgeting
to avoid overspending. In the soft drink example, the distributor
may submit a budget of $f(coke)=\$4$ and $f(sprite)=\$4$. This would
be suboptimal if he ends up exhausting his \$4 budget on coke but
spending only \$1 on coke. Needless to say, this would not be efficient
for the search engine either as additional revenue could be gained
from these potential ads display had the advertisers been given greater
flexibility. We believe that our matroid generalization may open up
window of opportunities to create values for both ad platforms
and advertisers.

Having motivated for the new matroid constraint, we formally define
the problems studied in this paper and our results. We note that many
of the other works in the literature also generalize readily to the
matroid setting but to avoid being repetitious, we have chosen only
a few representative problems to illustrate how we cope with the additional
matroid constraint without sacrificing performances.

\subsection{Our results and techniques}

We present $1-1/e$-competitive algorithms for all of the problems
considered in this paper. As we have alluded, our purpose is to illustrate
how the seemingly much more general matroid constraint can be incorporated
into existing works without any loss in competitive ratio. To this
end, we have selected two representative problems and show rigorously
that their matroid generalizations still admit $1-1/e$-competitive
algorithms. Specifically, the design and analysis of our new algorithms
make use of the charging scheme of \cite{Wang2013} and the primal-dual
frameworks of \cite{Buchbinder2007,devanurrandomized}.

To the best of our knowledge, this is the first time that the primal-dual
analysis is applied to an online problem which involves submodularity.
The closest example in online matching that we are aware of is~\cite{devanur2012online},
which involves \emph{continuous} concave functions rather than \emph{discrete}
submodular functions. We hope that the primal-dual analysis method
will emerge as a powerful tool for tackling submodular-flavored online
problems.

The new results
in this paper are summarized below.
\begin{itemize}
\item Optimal waterfilling $1-1/e$-competitive algorithm for Matroid
Online Bipartite Vertex Cover and Matching (in the adversarial model)
under the small bid assumption.
\item Greedy is $1-1/e$-competitive algorithm for Matroid Online Bipartite
Matching in the random arrival Model (without the small bid assumption).
\end{itemize}
All of the above algorithms are greedy in nature and hence simple
to implement. In addition to existing ideas in the literature, the
design and analysis of our algorithms introduce the following new
techniques and machineries.

\paragraph{Two dimensional charging scheme}

The recent work of~\cite{Wang2013} introduces a (single-dimensional)
charging scheme-based analysis for online bipartite vertex cover.
To tackle submodularity, we propose a new charging scheme that operates
on the so-called ``bar chart diagrams'' for the Lovasz extension.
The new scheme is based on a {\em two-dimensional} charging function
of the bar chart diagram. Our analysis is therefore much more involved
than the original scheme.

\paragraph{Convex programming duality}

The primal-dual framework \cite{Buchbinder2007,devanurrandomized}
is one of the most powerful techniques for attacking various online
matching problems, with almost all of the existing works based on
LP duality. To address matroid constraint one must however resort
to convex programming (or equivalently, LPs with exponentially many
variables or constraints). Our contribution is demonstrating that
the existing machineries for handling matroids and submodularity in
the offline optimization setting (e.g. Lovasz extension and base
polyhedrons) conform seamlessly to the online primal-dual framework.
The main issue here is understanding why the diminishing return property
of submodular functions is compatible with the algorithms for online
bipartite matching. Although it is hard to explain without referring
to the specific details of the algorithms, fundamentally everything
works out nicely in this paper because each ``linear component''
of the Lovasz extension is compatible with the online primal-dual
analysis. This point will become clear when we analyze our algorithms.

\subsection{Previous works}

Karp et al.~\cite{Karp1990} gave the optimal $1-1/e$-competitive
algorithm for the online bipartite matching problem. Different variants
of it have been extensively studied. These include $b$-matching~\cite{kalyanasundaram2000optimal},
vertex-weighted version~\cite{Aggarwal2011,devanurrandomized}, adwords~\cite{Buchbinder2007,DevenurH09,Mehta2007,devanurrandomized,devanur2012online,GoelM08,Aggarwal2011,devanur2012asymptotically,devenur2009adwords}
and online market clearing~\cite{Blum2006}.

To get around the worst-case analysis, other research directions study
the problem with weaker adversarial models by assuming stochastic
inputs~\cite{Feldman2009,Manshadi2011,Mahdian2011,Karande2011} as
well as general graphs~\cite{bansal2010lp}.

\paragraph{Organization}

The rest of this paper is structured as follows. In Section~\ref{sec:pre},
we introduce various notions and tools from combinatorial optimization
needed for our results. Section~\ref{sec:obvc} is on the basic online
bipartite vertex cover problem and introduces the charging scheme
from~\cite{Wang2013}. Section~\ref{sec:submodular} studies matroid
online bipartite vertex cover and builds on the previous charging
scheme. It also includes an alternate primal-dual analysis of the
algorithm which proves our result on matroid online bipartite matching.
Finally, Section~\ref{sec:Matroid-Online-Bipartite} is on matroid
online bipartite matching in the random arrival model. We conclude
in Section~\ref{sec:conclusion} with some open problems.

\section{Preliminaries}

\label{sec:pre} We first introduce some standard notions in combinatorial
optimization before stating our problems. In this paper we consider
only undirected bipartite graphs $G=(L,R,E)$, where $L$ and $R$
are the sets of left and right vertices respectively.

Given $G=(L,R,E)$, a {\em vertex cover} (VC) of $G$ is a subset
of vertices $C\subseteq L\cup R$ such that for each edge $(u,v)\in E$,
$C\cap\{u,v\}\neq\emptyset$. A {\em matching} of $G$ is a subset
of edges $M\subseteq E$ such that each vertex $u\in L,v\in R$ is
incident to at most one edge in $M$. The size of a vertex cover $C$
and a matching $M$ is just $|C|$ and $|M|$.

$(\by,\bz)\in([0,1]^{L},[0,1]^{R})$ is a {\em fractional vertex
cover} if for any edge $(u,v)\in E$, $y_{u}+z_{v}\geq1$. $\bx\in[0,1]^{E}$
is a {\em fractional matching} if for each $u\in L$, $x_{u}:=\sum_{v\in N(u)}x_{uv}\leq1$
and for each $v\in R$, $x_{v}:=\sum_{u\in N(v)}x_{uv}\leq1$. The
size of a fractional vertex cover $(\by,\bz)$ and a fractional matching
$M$ is just $\sum_{u\in L}y_{u}+\sum_{v\in R}z_{v}$ and $\sum_{e\in E}x_{e}=\sum_{u\in L}x_{u}=\sum_{v\in R}x_{v}$.
We call $y_{u}$ the {\em potential} of $u$. It is well-known
that VC and matching are dual of each other in bipartite graphs.

A set function $f:2^{L}\longrightarrow\mathbb{R}$ is said to be \textbf{submodular}
if for all $A,B\subseteq L$, 
\[
f(A)+f(B)\geq f(A\cup B)+f(A\cap B).
\]

One often finds the following equivalent definition useful: for every
$A,B\subseteq L$ with $A\subseteq B$ and every $e\in L$ (using the shorthand $A+e=A\cup\{e\}$), 
\[
f(A+e)-f(A)\geq f(B+e)-f(B).
\]

Loosely speaking, this says that the marginal return of adding an
element $e$ to a larger set is smaller. This property makes submodular
functions appealing beyond its mathematical beauty as this phenomenon
is observed in many real-life scenarios, especially those which arise
from economic settings.

A submodular function $f$ is \textbf{monotone} if $f(B)\geq f(A)$
for every $A,B\subseteq L$ with $A\subseteq B$.

Given a nonnegative\footnote{Our results actually still hold even if $f(S)<0$ for some $S\subseteq L$,
in which case we can just remove $S$ as its vertices can never be
matched.} monotone submodular function $f(\cdot)$, $\bx\in[0,1]^{E}$ is a
\textbf{matroid matching} defined by $f$ if for all $v\in R$, 
\[
x_{v}:=\sum_{u\in N(v)}x_{uv}\leq1,
\]
and for all $S\subseteq L$, 
\[
x_{S}:=\sum_{u\in S}x_{u}=\sum_{u\in S}\sum_{v\in N(u)}x_{uv}\leq f(S).
\]

It is easy to see that this is indeed a generalization of the usual
matching which corresponds to $f(S)=|S|$.

The rest of 
this section formally defines our problems and introduces the machineries needed to tackle our problems.

\subsection{Problem statements}

In all of these problems, the underlying graph is bipartite $G=(L,R,E)$
with offline vertices $L$. The right vertices $R$ arrive online
one by one. When a vertex $v\in R$ arrives, all of its incident edges
are revealed.
\begin{itemize}
\item \textbf{Online bipartite vertex cover (OBVC)} The algorithm must maintain
at all time a valid monotone vertex cover $C$, i.e. no vertex can
ever be removed from $C$ once it is put into $C$. Thus the algorithm
essentially decides whether to assign $v$ or $N(v)$ to $C$ upon
the arrival of an online vertex $v$. The objective is to minimize
the size of $C$ in the end.
\item \textbf{Matroid online bipartite vertex cover (MOBVC)} The setting
is exactly identical to OBVC except that the objective function is
$f(C\cap L)+|C\cap R|$ instead of $|C|$. Here $f(\cdot)$ is a {\em
nonnegative monotone submodular} function. OBVC is a special case
of MOBVC in which $f$ is simply the cardinality function $f(S)=|S|$.
\item \textbf{Matroid online bipartite (fractional) matching (MOBM)} The
setting of MOBM is similar to OBVC. Here we have a nonnegative monotone
submodular function $f(\cdot)$ on the left vertices and the algorithm
maintains a matroid matching $\bx$ instead of a vertex cover. When
an online vertex $v$ arrives, the algorithm must initialize all $x_{uv}$
for $u\in N(v)$ so that $\bx$ is still a valid matroid matching.
The objective is to maximize the size of $\bx$, i.e. $\sum_{e\in E}x_{e}=\sum_{u\in L}x_{u}=\sum_{v\in R}x_{v}$.
The integral version of MOBM has all $x_{e}\in\{0,1\}$ and $f$ being
a matroid rank function, which is integer-valued and nonnegative monotone submodular.
\end{itemize}
Finally, although not needed for this paper, we remark that the offline
versions of these problems can be solved in polynomial time by matroid
intersection.

%

\paragraph{Fractional vs. Integral}

As with the previous waterfilling algorithms for online bipartite
matching, the integral version with small bid assumption is equivalent
to the fractional version \cite{Buchbinder2007,devanur2013whole}.
For the sake of convenience we present our waterfilling algorithm
for MOBM in the fractional setting.

\subsection{Lovasz extension of submodular functions}

Given a submodular function $f:2^{L}\longrightarrow\mathbb{R}$, the
Lovasz extension $\hat{f}:[0,1]^{L}\longrightarrow\mathbb{R}$ is
a continuous convex relaxation of $f$ and is defined by 
\[
\hat{f}(\by)=\mathbb{E}_{t}[f(L(t))],
\]
where $L(t)=\{u\in L:y_{u}\geq t\}$ and the expectation is taken
over $t$ chosen uniformly at random from $[0,1]$. It is easy to
check one does have $f(S)=\hat{f}(I_{S})$. Here $I_{S}$ is the indicator
variable for $S\subseteq L$.

We will also make
heavy use of an equivalent definition in this paper. Given $\by\in[0,1]^{L}$,
order the vertices of $L=\{1,2,...,n\}$ in such a way that $0=y_{0}\leq y_{1}\leq y_{2}\leq...\leq y_{n} \leq y_{n+1} = 1$.
Let $Y_{i}=\{i,i+1,\ldots,n\}$ and $Y_{n+1}=\emptyset$. Then 
\[
\hat{f}(\by)=\sum_{i=1}^{n+1}(y_{i}-y_{i-1})f(Y_{i}).
\]

An immediate implication of this formulation is that by restricting
$\by\in[0,1]^{L}$ to some fixed ordering $\sigma:\{1,2,...,|L|\}\longrightarrow L$,
$\hat{f}(\by)$ is a linear function. This property will be used in
various places of this paper.

Finally, note that for monotone submodular function $f$, its Lovasz
extension $\hat{f}(\by)$ is monotonically increasing (in each coordinate).

In the next section, we introduce a bar chat interpretation of the
Lovasz function. This representation plays an important
role in analyzing MOBVC and MOBM.

\subsubsection{Bar-chart representation}

\label{sec:bar} Given $\by\in[0,1]^{L}$, the bar chart representation
of $\hat{f}(\by)$ is the set 
\[
\bigcup_{t\in[0,1]}\{t\}\times[0,f(L(t))].
\]
Notice that the bars are decreasing in height as $t$ increases because
$f$ is monotone. If we order $L=\{1,2,...,n\}$ in such a way that
$y_{1}\leq y_{2}\leq...\leq y_{n}$. Using the notation in the last
section, the bar chart representation consists of the bars 
\[
[0,y_{1}]\times[0,f(Y_{1})],[y_{1},y_{2}]\times[0,f(Y_{2})],...,[y_{n-1},y_{n}]\times[0,f(Y_{n})],[y_{n},1]\times[0,f(Y_{n+1})].
\]
This is often a useful way to visualize the Lovasz extension $\hat{f}(\by)=\sum_{i=1}^{n+1}(y_{i}-y_{i-1})f(Y_{i})$
as each term in the summand corresponds to precisely a bar in the
bar-chart representation. In particular, $\hat{f}(\by)$ is the area
of the bar chart.

Strictly speaking, a bar can be empty (e.g. when $y_{i}=y_{i+1}$)
but we shall implicitly disregard them hereafter as it does not affect
our proofs in any way and would only make the notations more cumbersome.

Readers may find that it is sometimes more intuitive to view the bar
chart as the function $t\mapsto f(L(t))$ for $t\in[0,1]$.
Figure \ref{fig:bar-chart} shows a bar chart
representation with 5 bars, of which the last one corresponds to the
empty set and has height $f(\emptyset)=0$.

\begin{figure}[h]
\centering \includegraphics[width=5cm,height=3cm]{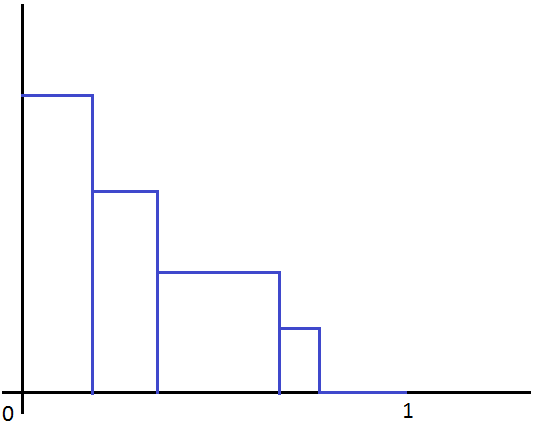} \protect\caption{Bar chart representation of Lovasz extension}
\label{fig:bar-chart} 
\end{figure}

\subsection{Convex program for bipartite matroid matching and vertex cover}

Recall the notations $x_{v}:=\sum_{u\in N(v)}x_{uv}$ and $x_{S}:=\sum_{u\in S}x_{u}:=\sum_{u\in S}\sum_{v\in N(u)}x_{uv}$.
The primal and dual convex programs below are used in the primal-dual
analysis of our algorithms for MOBM and MOBVC.

\begin{table}
       \tbl{}{%
       \begin{tabular}{|rl|rl|}
\hline 
Primal:  &  & Dual:  & \tabularnewline
 & $\max\sum_{e\in E}x_{e}$  &  & $\min\hat{f}(\by)+\sum_{v\in R}z_{v}$ \tabularnewline
s.t.  & $x_{v}\leq1,\,\forall v\in R$  & s.t.  & $y_{u}+z_{v}\geq1,\,\forall(u,v)\in E$ \tabularnewline
 & $x_{S}\leq f(S),\forall S\subseteq L$  &  & $\by,\bz\geq0$ \tabularnewline
 & $\bx\geq0$  &  & \tabularnewline
\hline 
       \end{tabular}}
       \label{label}
\end{table}


Readers who familiar with polymatroid intersection should recognize
that the primal is actually the polytope associated with the intersection
of a partition matroid on $R$ and a polymatroid on $L$ defined by
the submodular function $f$.

As in the usual primal-dual method, weak duality\footnote{In fact, even strong duality holds but this is not needed for our
analysis.} is required in order to bound the size of the primal and dual solutions.

\begin{lemma} \label{lem:weakdual} (weak duality) For any feasible
solutions $\bx$ and $(\by,\bz)$ to the primal and dual programs
above, we have 
\[
\sum_{e\in E}x_{e}\leq\hat{f}(\by)+\sum_{v\in R}z_{v}.
\]
\end{lemma}
  \begin{proof} Let $L(t)=\{u\in L:y_{u}\geq t\}$ and
define $R(1-t)$ analogously. For every $t\in[0,1]$, we claim that
\[
C(t):=L(t)\cup R(1-t)
\]
is a vertex cover of $G$. Consider any edge $(u,v)\in E$. If $y_{u}\geq t$
then $(u,v)$ is certainly covered. Otherwise, we have $z_{v}\geq1-y_{u}\geq1-t$
in which case $v\in R(1-t)$.

We are now ready to prove the lemma. For every vertex cover $C(t)$,
since there are no edges between $L\backslash L(t)$ and $R\backslash R(1-t)$,
we have 
\[
\sum_{e\in E}x_{e}\leq x_{L(t)}+\sum_{v\in R(1-t)}x_{v}\leq f(L(t))+|R(1-t)|.
\]

Our result then follows by noting that $\hat{f}(\by)=\mathbb{E}_{t}[f(L(t))]$
and $\sum_{v\in R}z_{v}=\mathbb{E}_{t}[|R(1-t)|]$, where the latter
equality holds because each $v\in R$ is chosen to be in $R(1-t)$
with probability $z_{v}$. \end{proof}

\subsection{Rounding scheme for online bipartite vertex cover}

\label{sec:rounding}

For a fractional vertex cover $(\by,\bz)$ in an online algorithm
for the online bipartite vertex cover problems studied in this paper,
we can always round it to an integral solution by following simple
scheme. We first sample $\gamma$ uniformly at random from $[0,1]$.
Afterwards, for any vertex $u\in L$, we place $u$ in the cover as
long as $y_{u}\ge\gamma$. On the other hand, for any vertex $v\in R$,
we place $v$ in the cover when $z_{v}\geq1-\gamma$. It is not hard
to verify that this rounding scheme indeed maintains a monotone vertex
cover since $(\by,\bz)$ is monotone.

Now consider an algorithm for the matroid online bipartite vertex
cover problem, with fractional solution $(\by,\bz)$. Let $C(\gamma)$
be the vertices in the integral cover after the rounding with $\gamma$.
The performance of our algorithm with this rounding scheme is 
\[
\mathbb{E}_{\gamma}[f(C(\gamma)\cap L)]+\mathbb{E}_{\gamma}[|C(\gamma)\cap R|]=\hat{f}(\by)+\sum_{v\in R}z_{v}.
\]
Therefore, this rounding scheme does not incur a loss for MOBVC. Since
OBVC is a special case of MOBVC, the rounding scheme also works for
OBVC.

\subsection{$\alpha$ and two integrals}
To simplify our notation, we denote the optimal competitive ratio
for OBVC as 
\[
1+\alpha:=\frac{1}{1-1/e}
\]
throughout this paper. In our analyses, the following two
integrals will often be useful.

\[
\int_{0}^{1}\frac{1-t}{t+\alpha}dt=\alpha,\;\;\;\;\;\;\;\;\int_{0}^{1}\frac{1}{t+\alpha}dt=1
\]

\section{Online Bipartite Vertex Cover}

\label{sec:obvc}

In this section, we present the algorithm $GreedyAlloction$ for OBVC
from~\cite{Wang2013} as well its charging-based analysis. They will
be the corner-stone of our new results.

\begin{algorithm}[h!]
\SetAlgoLined \protect\caption{$GreedyAllocation$}

Initialize for each $u\in L$, $y_{u}=0$\; \For{each
online vertex $v$} { $\max a\le1$ s.t. $(1-a)+\sum_{u\in N(v)}\max\{a-y_{u},0\}\leq1+\alpha$\;
Let $X=\{u\in N(v)|y_{u}<a\}$\; For each $u\in X$, $y_{u}\leftarrow a$\;
$z_{v}\leftarrow1-a$\;\
} 
\end{algorithm}

When an online vertex $v$ arrives, we can choose to place $v$ in
the cover which has a cost of 1. Alternately, we can put all the vertices
from $N(v)$ into the cover. In $GreedyAllocation$, we attempt to
put as much $N(v)$ into the cover with a resource constraint of $1+\alpha$.
$GreedyAllocation$ is greedy in the sense that we try to make $a$,
i.e. the potential on $N(v)$, as large as possible. 

Now we present the charging-based analysis of $GreedyAllocation$
from~\cite{Wang2013}. Let $C^{*}$ be a minimum vertex cover of
$G$. We will charge the potential increment to vertices of $C^{*}$
so that each vertex of $C^{*}$ is charged at most $1+\alpha$.

Given an online vertex $v$, we consider the following two cases.

(1) $v\in C^{*}$. In this case, we charge the potential increment
in $N(v)$ and $v$ in the algorithm to $v$. In particular, $v$
will be charged at most $1+\alpha$. 

(2) $v\notin C^{*}$ which implies $N(v)\subseteq C^{*}$. In this
case, $N(v)$ should take the potential increment on themselves
as well as $y_{v}=1-a$ used by $v$. The following charging scheme
is critical. Intuitively, if $\sum_{u\in X}(a-y_{u})=a+\alpha$,
we should charge $\frac{1-a}{a+\alpha}(a-y_{u})$ to $u\in X$ since
the fair ``unit charge'' is $\frac{1-a}{a+\alpha}$. Because $\frac{1-a}{a+\alpha}$
is decreasing in $a$, $\frac{1-a}{a+\alpha}(a-y_{u})$ can be upper
bounded by 
\[
\int_{y_{u}}^{a}\frac{1-t}{t+\alpha}dt.
\]

The next lemma indicates that the total charge is sufficient. \begin{lemma}
\label{lem:charging}\cite{Wang2013} Let $F(x)=\int_{0}^{x}\frac{1-t}{t+\alpha}dt$.
If $\sum_{u\in X}(a-y_{u})=a+\alpha$ for some set $X$ and $a\geq y_{u}$
for $u\in X$, then 
\[
1-a\leq\sum_{u\in X}\left(F(a)-F(y_{u})\right).
\]
\end{lemma}

The lemma implies that $v\notin C^{*}$ can be charged at most $1+F(1)-F(0)=1+\int_{0}^{1}\frac{1-t}{t+\alpha}dt=1+\alpha$.
By the previous discussion, we then have:

\begin{theorem} \label{thm:no alternation}\cite{Wang2013} $GreedyAllocation$
is $1+\alpha$-competitive for (fractional) OBVC. \end{theorem}

As mentioned in Section~\ref{sec:rounding}, it is possible to convert
any online fractional vertex cover algorithms while preserving the competitive ratio in expectation.

\section{Matroid Online Bipartite Matching and Vertex Cover}

\label{sec:submodular}

As in OBVC, we first consider the fractional version of MOBVC. Our
objective is then to minimize $\hat{f}(\by)+\sum_{v\in R}z_{v}$,
which is a convex relaxation of $f(C\cap L)+|C\cap R|$. At the end
of our analysis, we will show that it is possible to round our solution
to an integral VC with the same size in expectation. Thus our algorithm
for fractional MOBVC also works for integral MOBVC.

Our algorithm for MOBVC is still greedy. The analysis, however, relies
on a ''two-dimensional'' charging scheme in which the new additional
regions of the bar chart representation (introduced in Section~\ref{sec:bar})
are charged. We will see that the previous charging scheme for OBVC
is a simplistic version of this more sophisticated scheme.

We also give an alternate primal-dual analysis of our algorithm which
will imply a corresponding result for MOBM
as a by-product. Our method builds on the previous scheme~\cite{Buchbinder2007}
for online bipartite matching (and effectively OBVC).



Since the primal-dual analysis implies both the results on MOBM and
MOBVC, the charging analysis may seem redundant.
We stress that {\em both the charging-based and primal-dual analyses
are of interest}. Our charging-based analysis is very clean. It was
precisely for this reason that we were able to establish the result
on MOBVC first and ``reverse-engineer'' a primal-dual analysis which
is, in contrast, somewhat complicated. In retrospect, without the
charging analysis, we probably would not be able to come up with the
primal-dual analysis or even to realize that these problems admit
$1+\alpha$-approximation. Nonetheless, the primal-dual analysis is
still important since it implies an interesting result on MOBM.

\subsection{The algorithm}

The design of our algorithm for MOBVC is in the same spirit as OBVC.
In fact, the major modification needed is to replace $\sum_{u\in N(v)}\max\{a-y_{u},0\}$
by $\hat{f}(\by')-\hat{f}(\by)$ as now the objective function on
$L$ is $\hat{f}(\by)$ rather than $\sum_{u\in L}y_{u}$.

\begin{algorithm}[h!]
\SetAlgoLined \protect\caption{$GreedyAllocationSubmodular$}

Initialize for each $u\in L$, $y_{u}=0$\;\
\For{each online vertex $v$} { $\max a\le1$ s.t. $(1-a)+\hat{f}(\by')-\hat{f}(\by)\leq1+\alpha$,
where $y'_{u}=\max\{y_{u},a\}$ for $u\in N(v)$ and $y'_{u}=y_{u}$
for other $u$\; Let $X=\{u\in N(v)\mid y_{u}<a\}$\; For each $u\in X$,
$y_{u}\leftarrow a$\; $z_{v}\leftarrow1-a$\;\
} 
\end{algorithm}

The analysis of $GreedyAllocationSubmodular$ makes extensive use
of the bar chart representation introduced in Section~\ref{sec:bar}.
It is thus helpful to interpret our algorithm in terms of the bar
chart. This will hopefully also make the change in the Lovasz extension
$\hat{f}(\by')-\hat{f}(\by)$ more intuitive and easier to visualize.

\subsubsection{Bar chart interpretation of the algorithm}

We take a closer look at how the bar chart changes after processing
an online vertex. Recall that 
\[
L(t)=\{u\in L|y_{u}\geq t\}
\]
and $f(L(t))$ is the height of the bar chart at $t$. First of all,
observe that for the bars at $t\leq a$, the height changes from $f(L(t))$
to $f(L(t)\cup X)$ since the potential of the vertices from $X$
increased to $a$ and no other vertex increased in potential. As a
result, the bar at $t>a$ remains at the same height.

With this observation in mind, we see that a new rectangular region
(possibly empty) of height $f(L(t)\cup X)-f(L(t))$ is added to the
top of the bar at $t<a$. Moreover, the bar at $t=a$ is effectively
split into two\footnote{It is possible to have the degenerate case where $a$ coincides with
the boundary of a bar.}: the right one has the same height $f(L(a))$ whereas the left one
has a larger height $f(L(a)\cup X)\geq f(L(a))$.

Our charging scheme in the next section makes critical use of these two properties:
\begin{itemize}
\item All the new rectangular regions are added to the bars at $t\leq a$. 
\item The total area of the new rectangular regions is $\hat{f}(\by')-\hat{f}(\by)$. 
\end{itemize}
The mechanism in which $1-a$ is charged to $u\in X$ lies in the
heart of the previous charging scheme for OBVC. This idea does not
quite work anymore as our objective function is submodular rather
than modular. The key insight in our new analysis is to charge $1-a$
to the new rectangular regions of the bar chart. This is in contrast
to the previous scheme which charges to individual $u\in X$. Towards
the end of the next section, we will explain how it is actually a
simplistic version of our new scheme.

Our analysis in a nutshell is a careful study of figure \ref{fig:bar-split}.
The red regions are the new rectangles added to the bar chart. Note
that the first three bars increased in height with the third one being
split into two at $a$. All of the new regions are found at $t\leq a$.
It is no coincidence that the height of the red rectangles decreases
along the horizontal axis. Although not needed for the proof, it is
instructive to check that this phenomenon is an artifact of submodularity
and monotonicity.

In the next section, we propose a charging scheme in which the red
new regions are charged to compensate for $z_{v}=1-a$.

\begin{figure}[h]
\centering \includegraphics[width=5cm,height=3cm]{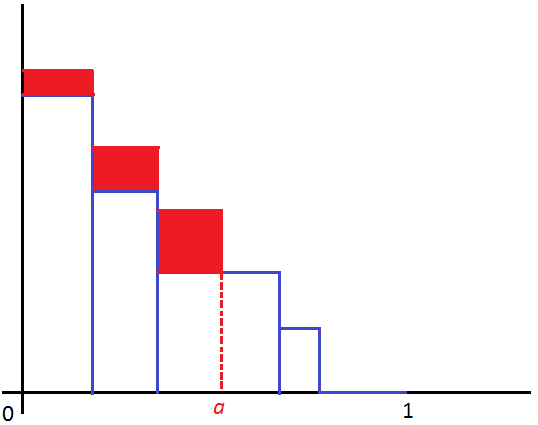} \protect\caption{Bar chart being split at $a$}

\label{fig:bar-split} 
\end{figure}

Finally, we remark that the bar chart is just a pictorial representation
of the Lovasz extension. We could have carried out the analysis without
it at the expense of added notational complexity. It is for the same
reason that various degenerate cases are deemphasized (e.g. we speak
of the bar at $t$ but $t$ can happen to be at the boundary between
two consecutive bars).

\subsection{Charging-based analysis}

When an online vertex not in the optimal cover is processed, we will
charge all the potential used on this vertex to its neighbors, which
must be in the optimal cover. More concretely, we charge the cost to
the bar chart representing $\hat{f}(\cdot)$. For each point $(x,y)$
of the bar chart, the charging density is $\frac{1-x}{x+\alpha}$.
We first show that such charging density is sufficient to account
for the potential of the online vertex.

\begin{lemma} Let $B$ and $B'$ be the bar charts before and after
processing online vertex $v$. Let $a$ be the final water-level on
the neighbors of $v$ after processing $v$. We have 
\[
\int_{B'\setminus B}\frac{1-x}{x+\alpha}\mathrm{d}A\geq1-a.
\]
\end{lemma} \begin{proof} The main idea is to charge $1-a$ to the
new region of the bar chart. From the discussion in the last section,
all the new regions have $x$-coordinates at most $a$. Therefore,
we have 
\begin{align}
\int_{B'\setminus B}\frac{1-x}{x+\alpha}\mathrm{d}A\geq\frac{1-a}{a+\alpha}\int_{B'\setminus B}\mathrm{d}A=\frac{1-a}{a+\alpha}\cdot(\hat{f}(\by')-\hat{f}(\by))=1-a,
\end{align}

where the last equality obviously holds if $a=1$. If $a<1$, then
we must have exhausted all of our resources $1+\alpha$ (otherwise
$a$ would be larger) and hence we have $\hat{f}(\by')-\hat{f}(\by)=a+\alpha$.
\end{proof}

Now we show that the total charges to the left vertices by online
vertices not in the optimal cover $C^{*}$ is at most $\alpha\cdot f(L\cap C^{*})$.

\begin{lemma} The total charges received from online vertices $R\setminus C^{*}$
is $\leq\alpha\cdot f(L\cap C^{*})$. \end{lemma}

\begin{proof} Let $B^{*}$ be the union of the new regions in the
bar chart generated by processing online vertices $R\setminus C^{*}$.
Therefore, the total charges are 
\[
\int_{B^{*}}\frac{1-x}{x+\alpha}\mathrm{d}A.
\]
For $t\in[0,1]$, let $B^{*}(t)$ be the intersection of $B^{*}$
with the line $x=t$. We have 
\begin{align*}
\int_{B^{*}}\frac{1-x}{x+\alpha}\mathrm{d}A=\int_{0}^{1}\int_{B^{*}(x)}\frac{1-x}{x+\alpha}\mathrm{d}y\mathrm{d}x\leq\int_{0}^{1}\frac{1-x}{x+\alpha}\mathrm{d}x\sup_{t\in[0,1]}\int_{B^{*}(t)}\mathrm{d}y=\alpha\cdot\sup_{t\in[0,1]}\int_{B^{*}(t)}\mathrm{d}y.
\end{align*}

It is then sufficient to show that for $t\in[0,1]$, 
\[
\int_{B^{*}(t)}\mathrm{d}y\leq f(L\cap C^{*}).
\]

Notice that $\int_{B^{*}(t)}\mathrm{d}y$ is the total height of
regions added to the bar chart at $x=t$ when processing online vertices
in $R\setminus C^{*}$.

Although the proof below looks somewhat technical, the key idea is
simple. Suppose that all of the vertices in $L\backslash C^{*}$ are
removed, i.e. $L\subseteq C^{*}$. Now the height of the bar chart
is at most $f(L)=f(C^{*}\cap L)$ so our claim is clear. If we add
back $L\backslash C^{*}$, recall that we care only about the rectangles
added for $v\notin C^{*}$. The height of the additional rectangle
is just the marginal difference, which cannot be worse than before
because of diminishing marginal return. We formalize this below.

Let $L_{i}(t)$ be the set $L(t)=\{u\in L\mid y_{u}\geq t\}$ \emph{after}
processing the $i$-th online vertex $v_{i}$. Since $y_{u}$ can
never decrease for all $u\in L$, we have 
\[
L_{0}(t)\subseteq L_{1}(t)\subseteq\cdots\subseteq L_{|R|}(t).
\]

Furthermore, $L_{i}(t)\backslash L_{i-1}(t)\subseteq N(v_{i})$ since
only $y_{u}$ for $u\in N(v_{i})$ can increase when processing $v_{i}$.
In particular, for $v_{i}\notin C^{*}$ we have that $L_{i}(t)\setminus L_{i-1}(t)\subseteq C^{*}$
as $v_{i}\notin C^{*}$ implies $N(v_{i})\subseteq C^{*}$. Submodularity
and $L_{i}(t)\setminus L_{i-1}(t)\subseteq C^{*}$ for $v_{i}\notin C^{*}$
give 
\begin{equation}
f(L_{i}(t))-f(L_{i-1}(t))\leq f(L_{i}(t)\cap C^{*})-f(L_{i-1}(t)\cap C^{*}).\label{eqn:local_sumodular}
\end{equation}

Finally, when processing $v_{i}$, the height of the new rectangular
region\footnote{Of course, it is possible that no region is added in which case this
is still okay as $f(L_{i}(t))=f(L_{i-1}(t))$.} at $t$ is precisely $f(L_{i}(t))-f(L_{i-1}(t))$. Now the sum of
the height of the rectangular regions at $t$ added when processing
$v_{i}\notin C^{*}$ is

\begin{align*}
\int_{B^{*}(t)}\mathrm{d}y & =\sum_{v_{i}\in R\setminus C^{*}}f(L_{i}(t))-f(L_{i-1}(t))\\
 & \leq\sum_{v_{i}\in R\setminus C^{*}}f(L_{i}(t)\cap C^{*})-f(L_{i-1}(t)\cap C^{*}) &  & \mbox{(submodularity)}\\
 & \leq\sum_{i=1}^{|R|}f(L_{i}(t)\cap C^{*})-f(L_{i-1}(t)\cap C^{*}) &  & \mbox{(monotonicity)}\\
 & =f(L_{|R|}(t)\cap C^{*})-f(L_{0}(t)\cap C^{*})\leq f(L\cap C^{*}).
\end{align*}

Here the last inequality follows from monotonicity and non-negativeness
of $f$.\end{proof}

\begin{lemma} The total resources used in processing online vertices
$R\setminus C^{*}$ are at most $(1+\alpha)\cdot f(L\cap C^{*})$.
\end{lemma} \begin{proof} For the $i$-th online vertex $v_{i}\in R$,
we define $\by_{i}$ to be the vector of potentials on $L$ {\em
after} processing $v_{i}$. Then, by our algorithm and the last lemma,
the total resources used in processing $R\setminus C^{*}$ are at
most 
\[
\alpha\cdot f(L\cap C^{*})+\sum_{v_{i}\in R\setminus C^{*}}\hat{f}(\by_{i})-\hat{f}(\by_{i-1}).
\]

Since for $v_{i}\in R\setminus C^{*}$, $L_{i}(t)\setminus L_{i-1}(t)\subseteq C^{*}$
for any $t\in[0,1]$, where $L_{i}(t)$ is defined as before. By Eqn.(\ref{eqn:local_sumodular})
and the definition of $\hat{f}(\cdot)$, we have 
\begin{align*}
\sum_{v_{i}\in R\setminus C^{*}}\hat{f}(\by_{i})-\hat{f}(\by_{i-1}) & \leq\sum_{v_{i}\in R\setminus C^{*}}\hat{f}(\by_{i}|_{L\cap C^{*}})-\hat{f}(\by_{i-1}|_{L\cap C^{*}})\\
 & \leq\sum_{v_{i}\in R}\hat{f}(\by_{i}|_{L\cap C^{*}})-\hat{f}(\by_{i-1}|_{L\cap C^{*}})\\
 & =\hat{f}(\by_{|R|}|_{L\cap C^{*}})-\hat{f}(0|_{L\cap C^{*}})\leq f(L\cap C^{*}),
\end{align*}
where $\by_{i}|_{L\cap C^{*}}$ restricts the vector $\by_{i}$ to
the vertices $L\cap C^{*}$ by setting the other entries to $0$.
This concludes the proof. \end{proof}

Therefore, our algorithm uses resources at most $(1+\alpha)\cdot f(L\cap C^{*})$
when processing vertices in $R\backslash C^{*}$. On the other hand,
it uses resources at most $(1+\alpha)\cdot|R\cap C^{*}|$ for other
online vertices as processing each of them increased the total potentials
by at most $1+\alpha$. Our algorithm thus $1+\alpha$-competitive
for the fractional matroid online bipartite vertex cover problem.
Since we can always round a fractional solution to a randomized integral
solution (section~\ref{sec:rounding}), we have the following theorem. 

\begin{theorem} There exists an optimal $1+\alpha$-competitive algorithm
for the matroid online bipartite integral vertex cover problem. \end{theorem}

\subsection{Primal-dual analysis}

We first review the key ingredients used in the original primal-dual
analysis of online bipartite matching in~\cite{Buchbinder2007},
which largely consists of two steps: 
\begin{itemize}
\item \textbf{Employs such constraints as $x_{u}=g(y_{u})$ (or $x_{u}\leq g(y_{u})$)
for some suitable increasing function $g$.} The motivation for doing
this is to enforce some correlation between the primal and dual variables
so that, for instance, when $x_{u}$ is small, $y_{u}$ is not too
big which allows room to pay for the future increase in $x_{u}$. 
\item \textbf{Relates the size of the primal and dual solutions by $\sum(g(a)-g(y_{u}))\approx c(1-a+\sum(a-y_{u}))$
for some constant $c$.} As in the usual primal-dual method, this
is essential for bounding the size of the solution via weak duality. 
\end{itemize}
This scheme depends crucially on the fact that the cost function is
modular. For submodular cost functions, one may try to imitate that
by using constraints like $x_{S}\leq f(S)g(h(\by|S))$ ($\by|S$ is
the vector restricted to $S$), where $g$ is the same as before and
$h:[0,1]^{S}\longrightarrow[0,1]$ is some suitable function.

Considering the Lovasz extension, the most natural choice is probably
$h(\by|S)=\min_{u\in S}y_{u}$. But this is fundamentally flawed as
one may have a very small $y_{u}$ with other $y_{u'}=1$. It turns
out that, perhaps somewhat counter-intuitively, the correct function
is $h(\by|S)=\max_{u\in S}y_{u}$.

Even more surprisingly, the constraint $x_{S}\leq f(S)g(h(\by|S))$
alone is not enough to relate the cost of the primal and dual solutions.
Recall that $\hat{f}(\by)=\sum f(Y_{i})(y_{i}-y_{i-1})$ for a fixed
ordering of $y$. Thus one might hope to consider $S=Y_{1},Y_{2},...$
in order to relate the increment in the size of the primal and dual
solutions. Unfortunately, this does not work as the ordering of $\by$
typically changes over the execution of the algorithm.

To rescue this, we turn to the bar chart representation again. Instead
of one \emph{global} ordering, a \emph{local} ordering is imposed
on each bar of the bar chart. More precisely, for a bar at $t$, we
maintain an ordering $\sigma_{t}$ of its existing vertices $L(t)$.
When $L(t)$ increases, we extend the current ordering by arbitrarily
appending the new vertices to its end. We formalize our ideas in the
rest of this section.

To simplify our notation, we view $x_{u}$ as a function on $[0,1]$
and the value of $x_{u}$\footnote{We abuse notations by using $x_{u}$ for both the primal variable as well as a function on $[0,1]$.} is 
\[
x_{u}=\int_{0}^{1}x_{u}(t)dt.
\]
This perspective will be useful when we analyze our algorithm using
the bar chart representation (which can be seen as a function on $[0,1]$).
The $x_{u}$ produced by the algorithm will be a piecewise constant
function. Conceptually, $\int_{0}^{1}x_{u}(t)dt$ aggregates over
the contribution of each bar to $x_{u}$.

\begin{algorithm}[h!]
\SetAlgoLined \protect\caption{$GreedyAllocationSubmodularPD$}

Initialize for each $u\in L$, $y_{u}=0,x_{u}(t)=0\forall t\in[0,1]$\;\
\For{each online vertex $v$} { Dual:\; $\max a\le1$ s.t. $(1-a)+\hat{f}(\by')-\hat{f}(\by)\leq1+\alpha$,
where $y'_{u}=\max\{y_{u},a\}$ for $u\in N(v)$ and $y'_{u}=y_{u}$
for other $u$\; Let $X=\{u\in N(v)\mid y_{u}<a\}$\; For each $u\in X$,
$y_{u}\leftarrow a$\; $z_{v}\leftarrow1-a$\; Primal:\; \For{each
bar of the bar chart at $[p,q]\ni t$ with a new rectangular region
$[p,q]\times[f(L(t)),f(L(t)\cup X)]$} { Extend the current ordering
$\sigma_{t}$ of $L(t)$ to $L(t)\cup X$ by appending $X\backslash L(t)$
arbitrarily to the end $\sigma_{t}(|L(t)|+1),...,\sigma_{t}(|L(t)\cup X|)$\;
For $t\in(p,q)$ and $|L(t)|+1\leq k\leq|L(t)\cup X|$, set 
\[
x_{\sigma_{t}(k)}(t)=\left(f\left(\bigcup_{i=1}^{k}\sigma_{t}(i)\right)-f\left(\bigcup_{i=1}^{k-1}\sigma_{t}(i)\right)\right)\frac{1}{a+\alpha}
\]
} For each $u\in N(v)$, set $x_{uv}$ to be the increment of $x_{u}=\int_{0}^{1}x_{u}(t)dt$
in this iteration\; } 
\end{algorithm}

At the first glance, our primal update seems somewhat convoluted.
The underlying philosophy is nevertheless much simpler. Before proceeding
to the analysis, we first unpack the details of the algorithm along
with some simple observations.

First of all, in our algorithm we focus on $x_{u}(t)$ rather than
$x_{uv}$. This is more convenient in the analysis since what matters
is the extent to which $u$ is matched (recall: $x_{S}\leq f(S)$)
but not which edge is assigned to $u$. Thus in the algorithm, we
determine only how much $x_{u}$ increases and retroactively what
$x_{uv}$ is.

Note that since each vertex can be added at most once to $L(t)$,
$x_{u}(t)$ can increase at most once and this increment will be from
$x_{u}(t)=0$ to $x_{u}(t)=\left(f\left(\bigcup_{i=1}^{k}\sigma_{t}(i)\right)-f\left(\bigcup_{i=1}^{k-1}\sigma_{t}(i)\right)\right)\frac{1}{a+\alpha}$,
where $u=\sigma_{t}(k)$.

Lastly, we emphasize the role of the ordering $\sigma_{t}$. This
is the key ingredient that makes the analysis possible. See Proposition
\ref{prop:xs} and Lemma \ref{lem:subconstraint} for more details.

We are now ready to analyze the algorithm. There are three major components:
\begin{itemize}
\item (feasibility) $x_{S}\leq f(S)$ for all $S\subseteq L$. 
\item (feasibility) $x_{v}\leq1$, i.e. the total increment of all $x_{u}$
in each iteration is at most 1.
\item (competitiveness) $\triangle D=(1+\alpha)\triangle P$, where $\triangle D$
and $\triangle P$ are the increments in the size of the dual and
primal solutions respectively. 
\end{itemize}
Once the above have been established, we can conclude that our algorithm
is correct and achieves a competitive ratio of $1+\alpha$ via weak
duality.

The following well-known property of the base polyhedron will
be used in the analysis. For completeness a proof is given here.

\begin{proposition} \label{prop:xs} Let $f:\Omega\longrightarrow\mathbb{R}$
be a monotone submodular function and fix an ordering $\sigma:\{1,2,...,|\Omega|\}\longrightarrow\Omega$.
Then the solution 
\[
x_{\sigma(k)}=f\left(\bigcup_{i=1}^{k}\sigma(i)\right)-f\left(\bigcup_{i=1}^{k-1}\sigma(i)\right)
\]
satisfies the inequalities $x_{S}\leq f(S)\forall S\subseteq\Omega$.
\end{proposition}
\begin{proof} Let $T_{j}=\cup_{i=1}^{j}\sigma(i)$.
Then $x_{\sigma(k)}=f(T_{k})-f(T_{k-1})$. For $S=\{s_{1},\ldots,s_{\ell}\}$,
\begin{align*}
x_{S} & =\sum_{i=1}^{\ell}x_{s_{i}}=\sum_{i=1}^{\ell}f(T_{\sigma^{-1}(s_{i})})-f(T_{\sigma^{-1}(s_{i})-1})\\
 & \leq\sum_{i=1}^{\ell}f(T_{\sigma^{-1}(s_{i})}\cap S)-f(T_{\sigma^{-1}(s_{i})-1}\cap S)\\
 & \leq\sum_{i=1}^{|\Omega|}f(T_{i}\cap S)-f(T_{i-1}\cap S)=f(S)-f(\emptyset)\leq f(S),
\end{align*}
where the first \& second inequalities follow from submodularity
and monotonicity. \end{proof}

\begin{lemma} \label{lem:subconstraint} In $GreedyAllocationSubmodularPD$,
we have $x_{S}\leq f(S)$ for all $S\subseteq L$. \end{lemma} \begin{proof}
We first show that for each $t$,
\[
x_{S}(t)\leq\frac{f(S)}{t+\alpha}
\]

Consider any $\sigma_{t}(k)\in S$ for which $x_{\sigma_{t}(k)}(t)>0$.
Then we must have set 
\[
x_{\sigma_{t}(k)}(t)=\left(f\left(\bigcup_{i=1}^{k}\sigma_{t}(i)\right)-f\left(\bigcup_{i=1}^{k-1}\sigma_{t}(i)\right)\right)\frac{1}{a+\alpha}\leq\left(f\left(\bigcup_{i=1}^{k}\sigma_{t}(i)\right)-f\left(\bigcup_{i=1}^{k-1}\sigma_{t}(i)\right)\right)\frac{1}{t+\alpha},
\]
where the inequality follows from the fact that only the bars on the
left of $a$ increase in height and hence $t\leq a$.

Now by Proposition~\ref{prop:xs}, we have $x_{S}(t)\leq\frac{f(S)}{t+\alpha}$.
Our desired result thus follows:

\[
x_{S}=\int_{0}^{1}x_{S}(t)dt\leq\int_{0}^{1}\frac{f(S)}{t+\alpha}dt=f(S).
\]
\end{proof}

\begin{lemma} \label{lem:boundpd} For each iteration of the algorithm,
the increases in the size of the primal and dual solutions satisfy
\[
\triangle D=(1+\alpha)\triangle P.
\]
\end{lemma}
\begin{proof} Recall that $\triangle D=\hat{f}(\by')-\hat{f}(\by)+1-a$
and $\hat{f}(\by')-\hat{f}(\by)$ is the total area of the new rectangular
regions needed to the bar chart.

On the other hand, $\triangle P$ is the sum of the increments of
all $x_{u}$. We restrict our attention to each bar via the following:

\begin{align*}
\sum_{u\in X}\Delta x_{u} & =\int_{0}^{1}\sum_{u\in X\setminus L(t)}x_{u}(t)dt=\int_{0}^{1}\sum_{k=|L(t)|+1}^{|L(t)\cup X|}x_{\sigma_{t}(k)}(t)\mathrm{d}t\\
 & =\int_{0}^{1}\sum_{k=|L(t)|+1}^{|L(t)\cup X|}\left(f\left(\bigcup_{i=1}^{k}\sigma_{t}(i)\right)-f\left(\bigcup_{i=1}^{k-1}\sigma_{t}(i)\right)\right)\frac{1}{a+\alpha}\mathrm{d}t\\
 & =\int_{0}^{1}\left(f(L(t)\cup X)-f(L(t))\right)\frac{1}{a+\alpha}\mathrm{d}t\\
 & =\frac{\hat{f}(\by')-\hat{f}(\by)}{a+\alpha},
\end{align*}

where the last equality holds as $f(L(t)\cup X)-f(L(t))$ is the height
of the new rectangular region at $t$.
In other words, 
\[
\triangle P=\frac{\hat{f}(\by')-\hat{f}(\by)}{a+\alpha}.
\]

The rest of the proof is now easy. The case $a=1$ is trivial as $\triangle D=\hat{f}(\by')-\hat{f}(\by)$.

If $a<1$, then we must have exhausted all of our resources $1+\alpha$.
Hence we have $\triangle D=1+\alpha$ and $\hat{f}(\by')-\hat{f}(\by)=a+\alpha$.
This gives $\triangle P=1$.
\end{proof}

\begin{corollary}
$x_v\leq 1$, i.e. the total increment of all $x_u$ in each iteration is $\leq1$.
\end{corollary}
\begin{proof}
The dual solution can increase by at most $1+\alpha$ and hence $\triangle P$, which is just $x_v$, is at most 1 by Lemma~\ref{lem:boundpd}.
\end{proof}

Combining all the pieces, we obtain our main theorem.

\begin{theorem} Our algorithm is $1-1/e$-competitive for matroid
online bipartite matching and $1+\alpha$-competitive for matroid
online bipartite vertex cover. \end{theorem} \begin{proof} By Lemma~\ref{lem:boundpd},
we always have $D=(1+\alpha)\cdot P$. By weak duality (see Lemma~\ref{lem:weakdual}),
we can bound $P$ and $D$ against the optimal solutions $D^{*}$
and $P^{*}$ as follows,

\[
P^{*}\leq D=(1+\alpha)\cdot P\leq(1+\alpha)\cdot D^{*}.
\]

This shows that $P\geq(1-1/e)P^{*}$ and $D\leq(1+\alpha)D^{*}$,
as desired. \end{proof}

Finally, we remark that we do have $x_{S}\leq f(S)\frac{\max_{u\in S}y_{u}+\int_{0}^{y_{u}}\frac{1-t}{t+\alpha}dt}{1+\alpha}$
(i.e. $x_{S}\leq f(S)g(\max_{u\in S}y_{u})$) as mentioned earlier.
Although not needed for the proof, it has served as a useful inspiration
when we were developing this primal-dual analysis.

\subsection{Extensions}

We briefly discuss extensions of our techniques to other problems
related to online bipartite matching. Using similar machineries, it
is rather straightforward to generalize the adwords \cite{Mehta2007,Buchbinder2007}
and online ad assignment problems to the matroid setting \cite{feldman2009online,devanur2013whole}.
We however do not discuss the details to avoid being repetitive.

For \textit{Matroid Adwords}, by considering the following primal
and dual programs one can argue that the waterfilling algorithm is
again $1-1/e$-competitive.

\begin{center}
\begin{tabular}{|rl|rl|}
\hline 
 & $\max\sum_{uv\in E}b_{uv}x_{uv}$  &  & $\min\hat{f}(\by)+\sum_{v\in R}z_{v}$ \tabularnewline
s.t.  & $x_{v}\leq1,\,\forall v\in R$  & s.t.  & $b_{uv}y_{u}+z_{v}\geq b_{uv},\,\forall(u,v)\in E$ \tabularnewline
 & $\sum_{u\in S}\sum_{v\in N(u)}b_{uv}x_{uv}\leq f(S),\forall S\subseteq L$  &  & $\by,\bz\geq0$ \tabularnewline
 & $\bx\geq0$  &  & \tabularnewline
\hline 
\end{tabular}
\par\end{center}

For matroid online ad assignment one would consider these
programs instead:

\begin{center}
\begin{tabular}{|rl|rl|}
\hline 
 & $\max\sum_{e\in E}w_{e}x_{e}$  &  & $\min\hat{f}(\by)+\sum_{v\in R}z_{v}$ \tabularnewline
s.t.  & $x_{v}\leq1,\,\forall v\in R$  & s.t.  & $y_{u}+z_{v}\geq w_{uv},\,\forall(u,v)\in E$ \tabularnewline
 & $x_{S}\leq f(S),\forall S\subseteq L$  &  & $\by,\bz\geq0$ \tabularnewline
 & $\bx\geq0$  &  & \tabularnewline
\hline 
\end{tabular}
\par\end{center}

\section{Matroid Online Bipartite Matching in the Random Arrival Model}
\label{sec:Matroid-Online-Bipartite}

It is known that Greedy is $1-1/e$-competitive for Online Bipartite
Matching in the random arrival model (without the small bid assumption)
\cite{GoelM08}. In this section, we combine the randomized primal-dual
analysis of \cite{devanurrandomized} with our convex program to prove
that Greedy remains $1-1/e$-competitive even for Matroid Online Bipartite
Matching. Recall that unlike the last section we are now working with
the integral version and $f$ is a matroid rank function.

\subsection{Review of Randomized Primal-Dual Analysis}

The elegant paper of \cite{devanurrandomized} introduced the randomized
primal-dual analysis, a clever yet simple extension of standard primal-dual.
One of the drawbacks of the standard primal-dual analysis is that
it is typically hard to accommodate for randomized algorithms, since
very often one has both dual feasibility and bounded duality gap at
all time. These patterns render primal-dual style analyses of most
randomized algorithms seemingly impossible. To get around with the
issue, randomized primal-dual requires only dual feasibility in expectation
(while still having bounded duality gap at all time). This is sufficient
because of linearity of expectation.

\subsection{Randomized Primal-Dual Analysis of Greedy for MOBM}

Greedy is a natural algorithm for (Matroid) Online Bipartite Matching
where we simply match a new online vertex $v$ to an available offline
vertex according to some fixed preference ordering $\sigma^{(v)}$.
Greedy is known to be $1-1/e$-competitive \cite{GoelM08} and in
this section, we generalize this result to MOBM.

Inspired by \cite{devanurrandomized}, we present a randomized primal-dual
analysis of Greedy. Let $M_{L}\subseteq L$ be the set of matched
vertices and $span(M_{L})=\{u\in L:f(M_{L}+u)=f(M_{L})\}$ be the
span of $M_{L}$ w.r.t. the matroid given by $f$. Note that a vertex
$u\in L$ can still be matched iff $u\notin span(M_{L})$, which we
therefore should keep track of.

Recall that we are working with the random arrival model where online
vertices arrive in a uniformly random order. Equivalently, we may
sample $t_{v}\in[0,1]$ for $v\in R$ and make them arrive in ascending
order prescribed by $t_{v}$. This parameterization is key to randomly
setting the dual variables $y_{u},z_{v}$. Let $g(t)=e^{t-1}$.

\begin{algorithm}[h!]
\SetAlgoLined \protect\caption{Greedy}

Initialize $M_{L}=\emptyset$ and for each $u\in L$,
$x_{u}=y_{u}=0$\;\

\For{each online vertex $v$}{

\textbf{Pass} if $N(v)\subseteq span(M_{L})$\; 

Let $u\in N(v)\backslash span(M_{L})$ be the first available vertex
in $\sigma^{(v)}$ \; 

$x_{uv}=1$\; 

$z_{v}=(1+\alpha)g(t_{v})$\;

For each $w\in span(M_{L}+u)\backslash span(M_{L})$, $y_{w}\leftarrow(1+\alpha)(1-g(t_{v}))$\; 

$M_{L}\longleftarrow M_{L}+u$\;

} 
\end{algorithm}

It is clear that the matching maintained by the algorithm is valid.
To prove that Greedy works, we need to establish (expected) dual feasibility
and duality gap.

Observe that by design each $w\in L$ is updated at most once. Moreover,
every $w\in span(M_{L})$ must have been updated.

\begin{lemma}[duality gap] 
For each iteration of the algorithm, the increases in the size of
the primal and dual solutions satisfy 
\[
\triangle D=(1+\alpha)\triangle P.
\]
\end{lemma} 

\begin{proof} 
If no vertex is matched, $\triangle D=\triangle P=0$ and the result
follows. Otherwise, $\triangle P=x_{uv}=1$ and we claim that $\triangle D=1+\alpha$.
Let $y'$ be the new $y$.

First of all, 
\[
\triangle D=z_{v}+\hat{f}(y')-\hat{f}(y)=(1+\alpha)g(t_{v})+\hat{f}(y')-\hat{f}(y)
\]
 so it suffices to show $\hat{f}(y')-\hat{f}(y)=(1+\alpha)(1-g(t_{v}))$.
We make two observations: $y'$ and $y$ differ exactly in $span(M_{L}+u)\backslash span(M_{L})$
and that for $w'\in span(M_{L})$, $y_{w'}\geq(1+\alpha)(1-g(t_{v}))$.
The former is clear while the latter follows from the fact that $w'$
was updated before $v$ arrives and hence the online vertex $v'$
used to update $w'$ must have 
\[
t_{v'}\leq t_{v}\implies y_{w'}=(1+\alpha)(1-g(t_{v'}))\geq(1+\alpha)(1-g(t_{v})).
\]

In other words, all the new coordinates in $y'$ have values not greater
than the existing ones. Now $\hat{f}(y')-\hat{f}(y)=(1+\alpha)(1-g(t_{v}))$
holds because $f(span(M_{L}+u))-f(span(M_{L}))=1$ (recall that $f$
is the rank function of the matroid).
\end{proof}

Next we show that the dual is feasible in expectation, i.e. $\mathbb{E}[y_{w}]+\mathbb{E}[z_{v}]\geq1$
for any edge $wv$. For this we need the notion of critical values.
Suppose that we run Greedy on $G$ without $v$ and let $w$ be updated
when an online vertex of random value $t^{c}\in[0,1]$ arrives (if
$w$ is never updated, $t^{c}=1$). We have the following two lemmas:

\begin{lemma}[dominance] 
Given $t_{v'}$ for $v'\neq v$, $v$ is matched whenever $t_{v}<t^{c}$.
\end{lemma} 

\begin{proof} 
If $t_{v}<t^{c}$, then $w$ must not have been updated when $v$
arrives. Furthermore, $w$ is not in the current $span(M_{L})$ as
it is updated only after in the run without $v$. Therefore when $v$
arrives, $w$ is available and $v$ can be matched (possibly to a vertex
other than $w$).
\end{proof}

\begin{lemma}[monotonicity] 
Given $t_{v'}$ for $v'\neq v$, $y_{w}\geq(1+\alpha)(1-g(t^{c}))$
regardless of the value of $t_{v}$.
\end{lemma} 

\begin{proof} 
Let $M_{L}$ and $M_{L}^{c}$ be the set of matched vertices in the
run with and without $v$ respectively. It is easy to show by induction
that $span(M_{L}^{c})\subseteq span(M_{L})$ at all time. In particular,
$w$ can only be updated earlier in the run with $v$. In other words,
the value $t$ used to update $w$ is either the same or smaller,
as desired.
\end{proof}

\begin{lemma}[dual feasibility] 
We have $\mathbb{E}[y_{w}]+\mathbb{E}[z_{v}]\geq1$ for any edge $wv$.
\end{lemma} 

\begin{proof} 
Let $t_{-v}$ denote the random values other than $t_{v}$. By the
previous two lemmas, $\mathbb{E}[z_{v}|t_{-v}]\geq\int_{0}^{t^{c}}(1+\alpha)g(t)dt$
and $\mathbb{E}[y_{w}|t_{-v}]\geq(1+\alpha)(1-g(t^{c}))$. Therefore
\[
\mathbb{E}[y_{w}|t_{-v}]+\mathbb{E}[z_{v}|t_{-v}]\geq(1+\alpha)(1-g(t^{c}))+\int_{0}^{t^{c}}(1+\alpha)g(t)dt=1
\]
where we used $g(t)=e^{t-1}$ for direct calculation. Now our result
follows by taking expectation over $t_{-v}$.
\end{proof}

Combining the previous lemmas, we obtain:

\begin{theorem} Greedy is $1-1/e$-competitive for matroid online
bipartite matching in the random arrival model. \end{theorem}

\section{Discussion and Future Research}

\label{sec:conclusion}

Our work naturally raises numerous questions about online algorithms.
We have selected three of them which we find the most interesting.
\begin{itemize}
\item \textbf{Submodularity in other online problems.} In online algorithms,
submodularity seems to be considered less often than other themes
in combinatorial optimization. We hope that our work will stimulate
the interest in combining submodularity with existing online problems
in the literature. Our results on MOBM and MOBVC
show that various ingredients used in offline submodular optimization
are still applicable online. We are hopeful that some of the powerful
machineries developed for submodularity over the past few
decades will find applications in various online settings.
\item \textbf{Two-sided submodular.} In both MOBM and MOBVC, a submodular
function $f$ is imposed on the left vertices $L$. A natural extension
is to impose also a submodular function $g$ on the right vertices
$R$. It is not hard to see that $1+\alpha$ cannot be attained anymore
for both problems. Nevertheless, the corresponding offline problems
are still solvable in polynomial time and do exhibit many nice typical
properties such as the existence of integral optimal solution. In
light of this, we believe that some constant approximation would still
be achievable.
\item \textbf{Integral MOBM.} Can we remove the small bid assumption needed
for the waterfilling algorithm? We conjecture that this algorithm
works for the special case of partition matroids: randomly permute
each partition and the vertices within each partition; match according
to this ordering. This algorithm generalizes the classical algorithm
RANKING \cite{Karp1990} in two ways: treating each offline vertex
as a singleton partition or treating the entire $L$ as one partition.
\end{itemize}

\bibliographystyle{ACM-Reference-Format-Journals}
\bibliography{online}


\begin{thebibliography}{00}


\ifx \showCODEN    \undefined \def \showCODEN     #1{\unskip}     \fi
\ifx \showDOI      \undefined \def \showDOI       #1{{\tt DOI:}\penalty0{#1}\ }
  \fi
\ifx \showISBNx    \undefined \def \showISBNx     #1{\unskip}     \fi
\ifx \showISBNxiii \undefined \def \showISBNxiii  #1{\unskip}     \fi
\ifx \showISSN     \undefined \def \showISSN      #1{\unskip}     \fi
\ifx \showLCCN     \undefined \def \showLCCN      #1{\unskip}     \fi
\ifx \shownote     \undefined \def \shownote      #1{#1}          \fi
\ifx \showarticletitle \undefined \def \showarticletitle #1{#1}   \fi
\ifx \showURL      \undefined \def \showURL       #1{#1}          \fi

\bibitem[\protect\citeauthoryear{Aggarwal, Goel, Karande, and Mehta}{Aggarwal
  et~al\mbox{.}}{2011}]%
        {Aggarwal2011}
{G. Aggarwal}, {G. Goel}, {C. Karande}, {and} {A. Mehta}. 2011.
\newblock \showarticletitle{Online vertex-weighted bipartite matching and
  single-bid budgeted allocations}. In {\em Proceedings of the Twenty-Second
  Annual ACM-SIAM Symposium on Discrete Algorithms}. SIAM, 1253--1264.
\newblock


\bibitem[\protect\citeauthoryear{Bansal, Gupta, Li, Mestre, Nagarajan, and
  Rudra}{Bansal et~al\mbox{.}}{2010}]%
        {bansal2010lp}
{N. Bansal}, {A. Gupta}, {J. Li}, {J. Mestre}, {V. Nagarajan}, {and} {A.
  Rudra}. 2010.
\newblock \showarticletitle{When LP is the cure for your matching woes:
  improved bounds for stochastic matchings}.
\newblock {\em Algorithms--ESA 2010\/} (2010), 218--229.
\newblock


\bibitem[\protect\citeauthoryear{Blum, Sandholm, and Zinkevich}{Blum
  et~al\mbox{.}}{2006}]%
        {Blum2006}
{A. Blum}, {T. Sandholm}, {and} {M. Zinkevich}. 2006.
\newblock \showarticletitle{Online algorithms for market clearing}.
\newblock {\em Journal of the ACM (JACM)\/} {53}, 5 (2006), 845--879.
\newblock


\bibitem[\protect\citeauthoryear{Buchbinder, Jain, and Naor}{Buchbinder
  et~al\mbox{.}}{2007}]%
        {Buchbinder2007}
{N. Buchbinder}, {K. Jain}, {and} {J. Naor}. 2007.
\newblock \showarticletitle{Online primal-dual algorithms for maximizing
  ad-auctions revenue}.
\newblock {\em Algorithms--ESA 2007\/} (2007), 253--264.
\newblock


\bibitem[\protect\citeauthoryear{Devanur and Jain}{Devanur and Jain}{2012}]%
        {devanur2012online}
{N.R. Devanur} {and} {K. Jain}. 2012.
\newblock \showarticletitle{Online matching with concave returns}. In {\em
  Proceedings of the 44th symposium on Theory of Computing}. ACM, 137--144.
\newblock


\bibitem[\protect\citeauthoryear{Devanur, Jain, and Kleinberg}{Devanur
  et~al\mbox{.}}{2013b}]%
        {devanurrandomized}
{N.R. Devanur}, {K. Jain}, {and} {R.D. Kleinberg}. 2013b.
\newblock \showarticletitle{Randomized Primal-Dual analysis of RANKING for
  Online Bipartite Matching}. In {\em SODA '13: Proceedings of the thirteenth
  Annual ACM-SIAM Symposium on Discrete Algorithms}.
\newblock
\newblock
\shownote{to appear.}


\bibitem[\protect\citeauthoryear{Devanur, Huang, Korula, Mirrokni, and
  Yan}{Devanur et~al\mbox{.}}{2013a}]%
        {devanur2013whole}
{Nikhil~R Devanur}, {Zhiyi Huang}, {Nitish Korula}, {Vahab~S Mirrokni}, {and}
  {Qiqi Yan}. 2013a.
\newblock \showarticletitle{Whole-page optimization and submodular welfare
  maximization with online bidders}. In {\em Proceedings of the fourteenth ACM
  conference on Electronic commerce}. ACM, 305--322.
\newblock


\bibitem[\protect\citeauthoryear{Devanur, Sivan, and Azar}{Devanur
  et~al\mbox{.}}{2012}]%
        {devanur2012asymptotically}
{Nikhil~R Devanur}, {Balasubramanian Sivan}, {and} {Yossi Azar}. 2012.
\newblock \showarticletitle{Asymptotically optimal algorithm for stochastic
  adwords}. In {\em Proceedings of the 13th ACM Conference on Electronic
  Commerce}. ACM, 388--404.
\newblock


\bibitem[\protect\citeauthoryear{Devenur and Hayes}{Devenur and Hayes}{2009a}]%
        {DevenurH09}
{Nikhil~R. Devenur} {and} {Thomas~P. Hayes}. 2009a.
\newblock \showarticletitle{The adwords problem: online keyword matching with
  budgeted bidders under random permutations}. In {\em EC '09: Proceedings of
  the tenth ACM conference on Electronic commerce}. ACM, New York, NY, USA,
  71--78.
\newblock
\showISBNx{978-1-60558-458-4}
\showDOI{%
\url{http://dx.doi.org/10.1145/1566374.1566384}}


\bibitem[\protect\citeauthoryear{Devenur and Hayes}{Devenur and Hayes}{2009b}]%
        {devenur2009adwords}
{Nikhil~R Devenur} {and} {Thomas~P Hayes}. 2009b.
\newblock \showarticletitle{The adwords problem: online keyword matching with
  budgeted bidders under random permutations}. In {\em Proceedings of the 10th
  ACM conference on Electronic commerce}. ACM, 71--78.
\newblock


\bibitem[\protect\citeauthoryear{Feldman, Korula, Mirrokni, Muthukrishnan, and
  P{\'a}l}{Feldman et~al\mbox{.}}{2009a}]%
        {feldman2009online}
{Jon Feldman}, {Nitish Korula}, {Vahab Mirrokni}, {S Muthukrishnan}, {and}
  {Martin P{\'a}l}. 2009a.
\newblock \showarticletitle{Online ad assignment with free disposal}.
\newblock In {\em Internet and Network Economics}. Springer, 374--385.
\newblock


\bibitem[\protect\citeauthoryear{Feldman, Mehta, Mirrokni, and
  Muthukrishnan}{Feldman et~al\mbox{.}}{2009b}]%
        {Feldman2009}
{J. Feldman}, {A. Mehta}, {V. Mirrokni}, {and} {S. Muthukrishnan}. 2009b.
\newblock \showarticletitle{Online stochastic matching: Beating 1-1/e}. In {\em
  Foundations of Computer Science, 2009. FOCS'09. 50th Annual IEEE Symposium
  on}. IEEE, 117--126.
\newblock


\bibitem[\protect\citeauthoryear{Goel and Mehta}{Goel and Mehta}{2008}]%
        {GoelM08}
{Gagan Goel} {and} {Aranyak Mehta}. 2008.
\newblock \showarticletitle{Online budgeted matching in random input models
  with applications to Adwords}. In {\em SODA '08: Proceedings of the
  nineteenth annual ACM-SIAM symposium on Discrete algorithms}. Society for
  Industrial and Applied Mathematics, Philadelphia, PA, USA, 982--991.
\newblock


\bibitem[\protect\citeauthoryear{Kalyanasundaram and Pruhs}{Kalyanasundaram and
  Pruhs}{2000}]%
        {kalyanasundaram2000optimal}
{B. Kalyanasundaram} {and} {K.R. Pruhs}. 2000.
\newblock \showarticletitle{An optimal deterministic algorithm for online
  b-matching}.
\newblock {\em Theoretical Computer Science\/} {233}, 1 (2000), 319--325.
\newblock


\bibitem[\protect\citeauthoryear{Karande, Mehta, and Tripathi}{Karande
  et~al\mbox{.}}{2011}]%
        {Karande2011}
{C. Karande}, {A. Mehta}, {and} {P. Tripathi}. 2011.
\newblock \showarticletitle{Online bipartite matching with unknown
  distributions}. In {\em Proceedings of the 43rd annual ACM symposium on
  Theory of computing}. ACM, 587--596.
\newblock


\bibitem[\protect\citeauthoryear{Karp, Vazirani, and Vazirani}{Karp
  et~al\mbox{.}}{1990}]%
        {Karp1990}
{R.M. Karp}, {U.V. Vazirani}, {and} {V.V. Vazirani}. 1990.
\newblock \showarticletitle{An optimal algorithm for on-line bipartite
  matching}. In {\em Proceedings of the twenty-second annual ACM symposium on
  Theory of computing}. ACM, 352--358.
\newblock


\bibitem[\protect\citeauthoryear{Mahdian and Yan}{Mahdian and Yan}{2011}]%
        {Mahdian2011}
{M. Mahdian} {and} {Q. Yan}. 2011.
\newblock \showarticletitle{Online bipartite matching with random arrivals: an
  approach based on strongly factor-revealing lps}. In {\em Proceedings of the
  43rd annual ACM symposium on Theory of computing}. ACM, 597--606.
\newblock


\bibitem[\protect\citeauthoryear{Manshadi, Gharan, and Saberi}{Manshadi
  et~al\mbox{.}}{2011}]%
        {Manshadi2011}
{V.H. Manshadi}, {S.O. Gharan}, {and} {A. Saberi}. 2011.
\newblock \showarticletitle{Online stochastic matching: Online actions based on
  offline statistics}. In {\em Proceedings of the Twenty-Second Annual ACM-SIAM
  Symposium on Discrete Algorithms}. SIAM, 1285--1294.
\newblock


\bibitem[\protect\citeauthoryear{Mehta, Saberi, Vazirani, and Vazirani}{Mehta
  et~al\mbox{.}}{2007}]%
        {Mehta2007}
{A. Mehta}, {A. Saberi}, {U. Vazirani}, {and} {V. Vazirani}. 2007.
\newblock \showarticletitle{Adwords and generalized online matching}.
\newblock {\em Journal of the ACM (JACM)\/} {54}, 5 (2007), 22.
\newblock
\showISSN{0004-5411}


\bibitem[\protect\citeauthoryear{Wang and Wong}{Wang and Wong}{2015}]%
        {Wang2013}
{Yajun Wang} {and} {Sam Chiu-wai Wong}. 2015.
\newblock \showarticletitle{Two-sided online bipartite matching and vertex
  cover: Beating the greedy algorithm}.
\newblock In {\em Automata, Languages, and Programming}. Springer, 1070--1081.
\newblock


\end{thebibliography}

\end{document}